\documentclass[journal]{IEEEtran}

\usepackage{color}
\usepackage{cite}
\usepackage{fixltx2e}
\usepackage[cmex10]{amsmath}
\usepackage{amssymb}
\usepackage{array}
\usepackage{wasysym}
\usepackage{dsfont}
\usepackage[latin1]{inputenc}                       
\usepackage[english]{babel}                         
\usepackage[T1]{fontenc}
\usepackage{mathtools}
\usepackage{amssymb} 
\usepackage{enumerate}
\usepackage{bbm}
\usepackage{epsfig,syntonly}
\usepackage{verbatim,times}
\usepackage{epstopdf}
\usepackage{graphicx}
\usepackage{latexsym,fancyhdr,bm}
\usepackage{subcaption}
\usepackage{url}
\usepackage{bbm}
\usepackage{dsfont}
\usepackage{amsthm}

\newtheoremstyle{custom}
{} 
{} 
{} 
{} 
{\bfseries} 
{:} 
{.25em} 
{} 
\theoremstyle{custom}

\newtheorem{theorem}{Theorem}
\newtheorem{lemma}{Lemma}
\newtheorem{proposition}{Proposition}
\newtheorem{definition}{Definition}
\newtheorem{example}{Example}

\newtheorem*{theorem*}{Theorem}
\newtheorem*{lemma*}{Lemma}
\newtheorem*{proposition*}{Proposition}
\newtheorem*{definition*}{Definition}
\newtheorem*{example*}{Example}
\newtheorem*{remark*}{Remark}
\newtheorem*{corollary*}{Corollary}

\makeatletter
\let\l@ENGLISH\l@english
\makeatother

\title{Construction of Polar Codes with Sublinear Complexity}

\author{Marco~Mondelli, S.~Hamed~Hassani, and~R\"{u}diger~Urbanke%
\thanks{M. Mondelli is with the Department of Electrical Engineering, Stanford University, USA
(e-mail: mondelli@stanford.edu).

R. Urbanke is with the School of Computer and Communication Sciences,
EPFL, CH-1015 Lausanne, Switzerland
(e-mail: ruediger.urbanke@epfl.ch).

S. H. Hassani is with the Department of Electrical and Systems Engineering, University of Pennsylvania, USA
(e-mail: hassani@seas.upenn.edu).}
}

\begin{document}

\maketitle
\begin{abstract}
\noindent Consider the problem of constructing a polar code of block length $N$ for the transmission over a given channel $W$. Typically this requires to compute the reliability of all the $N$ synthetic channels and then to include those that are sufficiently reliable. However, we know from \cite{Sc16, BDOT16} that there is a partial order among the synthetic channels. Hence, it is natural to ask whether we can exploit it to reduce the computational burden of the construction problem.

We show that, if we take advantage of the partial order \cite{Sc16, BDOT16}, we can construct a polar code by computing the reliability of roughly a fraction $1/\log^{3/2} N$ of the synthetic channels. In particular, we prove that  $N/\log^{3/2} N$ is a lower bound on the number of synthetic channels to be considered and such a bound is tight up to a multiplicative factor $\log\log N$. This set of roughly $N/\log^{3/2} N$ synthetic channels is universal, in the sense that it allows one to construct polar codes for any $W$, and it can be identified by solving a maximum matching problem on a bipartite graph.

Our proof technique consists in reducing the construction problem to the problem of computing the maximum cardinality of an antichain for a suitable partially ordered set. As such, this method is general and it can be used to further improve the complexity of the construction problem in case a new partial order on the synthetic channels of polar codes is discovered. 
\end{abstract}

\begin{IEEEkeywords}
Polar codes; partial order; construction problem; antichain; chain.
\end{IEEEkeywords}

\section{Introduction} \label{sec:intro}

Polar codes, introduced by Ar{\i}kan \cite{Ari09}, achieve the
capacity of any binary memoryless symmetric (BMS) channel with
encoding and decoding complexity $\Theta(N \log_2 N)$, where $N$
is the block length of the code. A unified characterization of the
performance of polar codes in several regimes is presented in
\cite{MHU15unif-ieeeit}. Let us mention the following basic facts:
the error probability scales with the block length roughly as
$2^{-\sqrt{N}}$ \cite{ArT09}; the gap to capacity scales with the
block length as $N^{-1/\mu}$, and bounds on the scaling exponent
$\mu$ are provided in \cite{HAU14, GB14, MHU15unif-ieeeit}; and
polar codes are not affected by error floors \cite{MHU15unif-ieeeit}.
A successive cancellation list (SCL) decoder with space complexity
$O(L N)$ and time complexity $O(L N \log_2 N)$ is proposed in
\cite{TVa15}, where $L$ is the size of the list. Empirically, the
use of several concurrent decoding paths yields an error probability
comparable to that under optimal MAP decoding with practical values
of the list size. Furthermore, by adding only a few extra bits of
cyclic redundancy check (CRC) precoding, the resulting performance
is comparable with state-of-the-art LDPC codes. Because of their
attractive features, polar codes are being considered for use in
future wireless communication systems (e.g., 5G cellular systems).

The idea of channel polarization is to take independent copies of
the transmission channel and to transform them into a set of reliable
channels and a set of unreliable channels, in such a way that the
overall capacity is preserved. Then, the information bits are
transmitted in the positions corresponding to the reliable channels
and the remaining positions are \emph{frozen} (i.e., their value
is shared between the encoder and the decoder). Therefore, in order
to construct a polar code, we need to identify the positions
corresponding to the reliable synthetic channels. Several techniques
have proposed to estimate the reliability of the synthetic channels:
Monte Carlo simulations \cite{Ari09}, density evolution
\cite{MoT09,MoTCommLett09}, Gaussian approximation of density
evolution \cite{Tr12}, efficient degrading and upgrading methods
\cite{TV13con, RHTT}.

In general, the ranking of the synthetic channels depends on the
specific transmission channel. Hence, one solution to the problem
of code construction is to evaluate the reliability of \emph{all}
synthetic channels. However, it was observed that there is a partial
order between the synthetic channels, which holds for any transmission
channel. A first partial order was described in \cite{MoTCommLett09}
and it was combined with a different partial order in the two
independent works \cite{Sc16, BDOT16}. In \cite{Sc16}, it is also
empirically shown that, by exploiting this combined partial order,
the complexity of the code construction can be significantly reduced.

In this paper, we give a tight characterization of the complexity
reduction guaranteed by the exploitation of this partial order. In
particular, we derive \emph{universal} bounds on the number of
synthetic channels whose reliability has to be computed in order
to construct the polar code. The bounds are \emph{universal} in the
sense that they hold for any transmission channel. Our main result
consists in proving an upper and a lower bound that differ by a
factor of $\log\log N$, where $N=2^n$ is the block length of the
code. The lower bound is equal to a known integer sequence, i.e.,
the maximal number of subsets of $\{1,2,\ldots,n\}$ that share the
same sum (sequence A025591 in \cite{OEISref}). Such a sequence scales as $N/\log^{3/2}
N$, which means that we need to compute the reliability of roughly
a fraction $1/\log^{3/2} N$ of the synthetic channels. In other words,
in order to construct a polar code, it suffices to know the reliability
of a sublinear number of synthetic channels.  In practice, this means that
at moderate block lengths ($N\approx 10^3$) we can save at least 80\% of the
channel computations. 

The remainder of this paper is organized as follows. In Section
\ref{sec:prel}, we set up the notation, describe the partial
order derived in \cite{Sc16, BDOT16}, and formalize the construction problem. In Section \ref{sec:main}, we state the main result about the complexity of the construction problem and we present its immediate implications. In Section \ref{sec:proof}, we give the proof and we describe how to actually find the channels whose reliability has to be computed. In Section \ref{sec:concl}, we provide some concluding remarks. The proofs of some intermediate results are deferred to the Appendix.

\section{Preliminaries} \label{sec:prel}

\subsection{Reliability Measures and Degradation}\label{subsec:reliable}

Let $W : \mathcal X \to \mathcal Y$ be a BMS channel with input alphabet $\mathcal X= \{0, 1\}$, output alphabet $\mathcal Y$, and transition probabilities $p_{Y\mid X}(y \mid x)$ for $x\in \mathcal X$ and $y\in \mathcal Y$. The random variables representing the input and the output of the channel are denoted by $X$ and $Y$, respectively. Since the channel is symmetric, we impose a uniform prior on the input, i.e., $p_X(0) = p_X(1)=1/2$.

There are several measures of the reliability of a channel, as specified by the following definition. 

\begin{definition}[Reliability Measures]\label{def:reliable}
Let $W : \mathcal X \to \mathcal Y$ be a BMS channel with transition probabilities $p_{Y\mid X}(y \mid x)$ for $x\in \mathcal X$ and $y\in \mathcal Y$. The reliability of $W$ is measured by one of the following quantities:
\begin{itemize}
	\item The \emph{mutual information} $I(W)$, defined as
\begin{equation}
	I(W) = I(X;Y);
\end{equation}

	\item The \emph{Bhattacharyya parameter} $Z(W)$, defined as
\begin{equation}
Z(W) = \sum_{y\in \mathcal Y} \sqrt{p_{Y\mid X}(y\mid 0)p_{Y\mid X}(y\mid 1)};
\end{equation}

	\item The \emph{MAP error probability} $P_{\rm e}(W)$, defined as
\begin{equation}
P_{\rm e}(W) = {\mathbb P}(X\neq \hat{x}(Y)),
\end{equation}
where $\hat{x}(y) = {\rm argmax}_x  p_{X\mid Y}(x\mid y)$ is the MAP decision of $X$ given $Y$.	

\end{itemize}

\end{definition}

Note that a channel is \emph{reliable} when it has a \emph{large} mutual information, a \emph{small} Bhattacharyya parameter, and a \emph{small} MAP error probability.

Let us now define the concept of stochastic degradation.

\begin{definition}[Stochastic Degradation]
Let $W_1 : \mathcal X \to \mathcal Y_1$ and $W_2 : \mathcal X\to Y_2$ be two BMS channels with respective transition probabilities $p_{Y_1\mid X}(y_1\mid x)$ and $p_{Y_2\mid X}(y_2\mid x)$, for $x\in \mathcal X$, $y_1 \in \mathcal Y_1$, and $y_2\in \mathcal Y_2$. We say that $W_1$ is \emph{stochastically degraded} with respect to $W_2$ and we write $W_1 \preceq W_2$ if there exists a memoryless channel with transition probabilities $p_{Y_1\mid Y_2}(y_1\mid y_2)$ such that for all $x\in \mathcal X$ and $y_1 \in \mathcal Y_1$,
\begin{equation}
p_{Y_1\mid X}(y_1\mid x) = \sum_{y_2\in \mathcal Y_2}p_{Y_1\mid Y_2}(y_1\mid y_2)p_{Y_2\mid X}(y_2\mid x).
\end{equation}
\end{definition}

If a channel is stochastically degraded, all the reliability measures defined in Definition \ref{def:reliable} become worse. This means that the mutual information decreases, the Bhattacharyya parameter increases, and the error probability increases. Such a fact is formalized by the following proposition (see Theorem 4.76 of \cite{RiU08} or Lemma 3 of \cite{TV13con}). 

\begin{proposition}[Stochastic Degradation and Reliability Measures]\label{prop:degrrel}
Let $W_1 : \mathcal X \to \mathcal Y_1$ and $W_2 : \mathcal X\to Y_2$ be two BMS channels and assume that $W_1\preceq W_2$. Then,
\begin{align}
I(W_1) &\le I(W_2),\\
Z(W_1) &\ge Z(W_2),\\
P_{\rm e}(W_1) &\ge P_{\rm e}(W_2).
\end{align}
\end{proposition}

\subsection{Synthetic Channels}

The basis of channel polarization consists in mapping two identical copies of the channel $W: \mathcal{X}\to \mathcal{Y}$ into the pair of channels $W^0: \mathcal{X}\to \mathcal{Y}^2$ and $W^1:\mathcal{X}\to \mathcal{X}\times\mathcal{Y}^2$, defined as
\begin{align}
W^0(y_1, y_2\mid x_1) & = \sum_{x_2\in \mathcal X} \frac{1}{2}W(y_1\mid x_1 \oplus x_2) W(y_2\mid x_2),\label{eq:minus}\\
W^1(y_1, y_2, x_1\mid x_2) & = \frac{1}{2}W(y_1\mid x_1 \oplus x_2) W(y_2\mid x_2).\label{eq:plus}
\end{align}
Then, $W^0$ is a worse channel in the sense that it is degraded with respect to $W$, hence less reliable than $W$; and $W^1$ is a better channel in the sense that it is upgraded with respect to $W$, hence more reliable than $W$.

By iterating this operation $n$ times, we map $N=2^n$ identical copies of the transmission channel $W$ into the synthetic channels $\{W_N^{(i)}\}_{i\in \{0, \ldots, N-1\}}$. More specifically, given $i\in \{0, \ldots, N-1\}$, let $(i_1, i_2, \ldots, i_n)$ be its binary expansion over $n$ bits, where $i_1$ is the most significant bit and $i_n$ is the least significant bit, i.e.,
\begin{equation}\label{eq:defbinexp}
i = \sum_{k=1}^n i_k 2^{n-1-k}.
\end{equation} 
Then, we define the synthetic channels $\{W_N^{(i)}\}_{i\in \{0, \ldots, N-1\}}$ as
\begin{equation}\label{eq:syntchan}
W_N^{(i)} = (((W^{i_1})^{i_2})^{\cdots})^{i_n}.
\end{equation}

\begin{example}[Synthetic Channel]
Take $n=4$ and $i=10$. Then, the synthetic channel $W_{16}^{(10)} = (((W^{1})^{0})^{1})^{0}$ is obtained by applying first \eqref{eq:plus}, then \eqref{eq:minus}, then \eqref{eq:plus}, and finally \eqref{eq:minus}.
\end{example}




\subsection{Partial Order}

In order to describe the partial order, it is helpful to define two operators, i.e., the addition and the left-swap operator, that map the index of a synthetic channel into the index of another synthetic channel. 

\begin{definition}[Addition Operator]\label{def:add}
Let $i\in \{0, \ldots, N-1\}$ and denote by $(i_1, i_2, \cdots, i_n)$ its binary expansion over $n$ bits, defined in \eqref{eq:defbinexp}. Given $k\in \{1, \ldots, n\}$, the \emph{addition operator} at position $k$ maps $i$ into $A^{(k)}(i)\in \{0, \ldots, N-1\}$. The binary expansion over $n$ bits of $A^{(k)}(i)$ is defined as 
\begin{equation}
(A^{(k)}(i))_{\ell} = 
\begin{cases}
1, & \ell=k, \\
i_\ell, & \ell \neq k. 
\end{cases}
\end{equation}
\end{definition}

In words, the addition operator $A^{(k)}$ takes the input $i$ and sets to $1$ the $k$-th of its binary expansion. Note that, if $i_k=1$, the addition operation $A^{(k)}$ simply copies the input into the output.

\begin{example}[Addition Operator]\label{ex:add}
Take $n=4$ and $i =10$. Note that $i$ has binary expansion $(1, 0, 1, 0)$. Then, $A^{(2)}(10)=14$ and its binary expansion is $(1, 1, 1, 0)$. Furthermore, $A^{(3)}(10)=10$ and its binary expansion is $(1, 0, 1, 0)$.
\end{example}

\begin{definition}[Left-swap Operator]\label{def:lw}
Let $i\in \{0, \ldots, N-1\}$ and denote by $(i_1, i_2, \cdots, i_n)$ its binary expansion over $n$ bits, defined in \eqref{eq:defbinexp}. Given $k\in \{2, \ldots, n\}$, the \emph{left-swap operator} at position $k$ maps $i$ into $L^{(k)}(i)\in \{0, \ldots, N-1\}$. If $i_k \neq 1$ {\em or} $i_{k-1} \neq 0$ then $L^{(k)}(i)=i$. Otherwise, the binary expansion over $n$ bits of $L^{(k)}(i)$ is defined as
\begin{equation}
(L^{(k)}(i))_\ell = 
\begin{cases}
1, & \ell=k-1, \\
0, & \ell=k, \\
i_\ell, & \ell \not \in \{k-1, k\}. 
\end{cases}
\end{equation}
\end{definition}

In words, the left-swap operator $L^{(k)}$ takes the input and, if possible, it swaps the $1$ in the $k$-th position with the bit on its left. This means that, if $i_k = 1$ {\em and} $i_{k-1} = 0$, the left-swap operator $L^{(k)}$ swaps position $k-1$ with position $k$. Otherwise, it simply copies the input into the output.

\begin{example}[Left-swap Operator]\label{ex:shift}
Take $n=4$ and $i =10$. Note that $i$ has binary expansion $(1, 0, 1, 0)$. Then, $L^{(2)}(10)=10$ and its binary expansion is $(1, 0, 1, 0)$. Furthermore, $L^{(3)}(10)=12$ and its binary expansion is $(1, 1, 0, 0)$.
\end{example}

We are now ready to describe the partial order introduced in \cite{Sc16, BDOT16}.

\begin{proposition}[Partial Order]\label{prop:order}
Let $W$ be a BMS channel and consider the $N=2^n$ synthetic channels $\{W_N^{(i)}\}_{i\in \{0, \ldots, N-1\}}$ obtained from $W$ by applying \eqref{eq:syntchan}. Then, for any $i \in \{0, \ldots, N-1\}$, 
\begin{align}
W_N^{(i)} & \preceq W_N^{(A^{(k)}(i))}, \quad \forall \hspace{0.15em}k\in \{1, \ldots, n\},\label{eq:add} \\
W_N^{(i)} & \preceq W_N^{(L^{(k)}(i))}, \quad \forall \hspace{0.15em}k\in \{2, \ldots, n\}.\label{eq:shift}
\end{align}
\end{proposition}

For the proof of \eqref{eq:add}, see Section V of \cite{MoTCommLett09} and, for the proof of \eqref{eq:shift}, see Theorem 1 of \cite{Sc16}. Note also that the partial order \eqref{eq:shift} was first implicitly pointed out in \cite{ArT09}. Furthermore, observe that \eqref{eq:add} and \eqref{eq:shift} hold for any BMS channel $W$. For this reason, we say that the partial order of Proposition \ref{prop:order} is \emph{universal}. 

\begin{example}[Partial Order]
Take $n=4$ and $i =10$. By applying Proposition \ref{prop:order} and recalling Examples \ref{ex:add} and \ref{ex:shift}, we immediately conclude that $W_{16}^{(10)}\preceq W_{16}^{(12)}$ and $W_{16}^{(10)}\preceq W_{16}^{(14)}$.  
\end{example}


\subsection{Construction Problem}

Given a BMS channel $W$ and a block length $N$, the problem of the construction of polar codes consists in selecting the set of the most reliable synthetic channels defined as in \eqref{eq:syntchan}. According to Definition \ref{def:reliable}, there are several notions of reliability. Since all these reliability measures become worse under stochastic degradation by Proposition \ref{prop:degrrel}, it does not really matter which one we choose. To fix the ideas, let us consider the Bhattacharyya parameter and define the construction problem to be the selection of the set of synthetic channels with the lowest Bhattacharyya parameters. However, keep in mind that the arguments and the conclusions of this paper remain valid when we choose the mutual information or the MAP error probability as reliability measures. Indeed, in this paper we exploit the partial order of Proposition \ref{prop:order}, which is an ordering of the synthetic channels in the sense of the stochastic degradation. 

\begin{definition}[Construction Problem]\label{def:constr}
Let $W$ be a BMS channel and consider the $N=2^n$ synthetic channels $\{W_N^{(i)}\}_{i\in \{0, \ldots, N-1\}}$ obtained from $W$ by applying \eqref{eq:syntchan}. In order to construct a polar code of block length $N$, we need to solve either the {\em fixed rate (FR) problem} or the {\em fixed performance (FP) problem} that are defined as follows.
\begin{itemize}

\item \emph{Fixed rate (FR) problem.} Given a block length $N$ and a rate $R\in (0, 1)$,
output the set of $\lfloor N R \rfloor$ synthetic channels with the smallest Bhattacharyya parameters. 

\item \emph{Fixed performance (FP) problem.} Given a block length $N$ and a
threshold $\gamma\in (0,1)$, output all the synthetic channels
whose Bhattacharyya parameter is smaller than $\gamma$.

\end{itemize}
\end{definition}

In the sequel we will limit our discussion to the FP construction problem. Note that if we can solve the FP construction problem, we can also solve the FR construction problem by simply performing a bisection on the values of the threshold.

\section{Statement of the Main Result} \label{sec:main}

\begin{theorem}[Complexity of FP Construction Problem]\label{th:main}
Let $W$ be a BMS channel and $N=2^n$ be the block length. Let $M(n)$ be the maximal number of subsets of $\{1,\ldots,n\}$ that share the same sum. Consider the partial order of Proposition \ref{prop:order} and use it to solve the FP construction problem with threshold $\gamma$ of Definition \ref{def:constr}. Then, the complexity of such a task can be bounded as follows.

\begin{itemize}
\item {\em Upper bound:} it suffices to compute the Bhattacharyya parameter of \emph{at most} $$M(n) \cdot \log \left(\frac{2^{n+1}}{M(n)}\right)$$ synthetic channels, for any $\gamma\in (0, 1)$.

\item {\em Lower bound:} it is necessary to compute the Bhattacharyya parameter of \emph{at least} $M(n)$ synthetic channels, for some $\gamma\in (0, 1)$.
\end{itemize}
\end{theorem}

The upper and the lower bounds provided by Theorem \ref{th:main} are represented in Figure \ref{fig:bounds} for $n\in \{6, 7, \ldots, 24\}$.

Let us point out that the results above are \emph{universal} in the sense that they hold for any BMS channel $W$. Note also that the upper bound holds for any choice of the threshold $\gamma\in (0, 1)$. On the contrary, the lower bound holds for some $\gamma\in (0, 1)$. Indeed, for some specific values of the threshold, the task might be easier. For example, if $\gamma$ is very small (e.g. $\gamma< Z(W)^N$), then none of the synthetic channels have a Bhattacharyya parameter smaller than $\gamma$. Similarly, if $\gamma$ is very large (e.g., $\gamma > 1-(1-Z(W))^N$), then all the synthetic channels have a Bhattacharyya parameter smaller than $\gamma$. Hence, it is interesting to provide a lower bound for the ``hard'' instances of $\gamma$.

\begin{figure}
\includegraphics[width = \columnwidth]{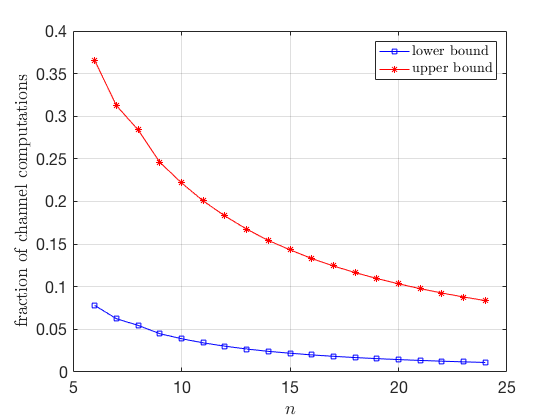}
\caption{Upper and lower bound on the fraction of channels whose Bhattacharyya parameter has to computed in order to solve the FP construction problem.}\label{fig:bounds}
\vspace{-.5cm}
\end{figure}

The sequence $M(n)$ is the integer sequence A025591 in \cite{OEISref}. The following lemma, stated below and proved in Appendix \ref{app:asym}, provides an asymptotic formula for it. 

\begin{theorem}[Asymptotic Formula for $M(n)$]\label{th:asym}
Let $M(n)$ be the maximal number of subsets of $\{1,\ldots,n\}$ that share the same sum. Then,
\begin{equation}\label{eq:asympt}
M(n) = \sqrt{\frac{6}{\pi}}\frac{2^n}{n^{3/2}}(1 +o(1)).
\end{equation}
\end{theorem}

By applying Theorem \ref{th:main} and \ref{th:asym}, we immediately conclude that, in order to solve the FP construction problem, we need to compute the Bhattacharyya parameter of roughly $N/\log^{3/2} N$ synthetic channels. Furthermore, the upper and the lower bound differ by a multiplicative factor of $\log (2N/M(n))$, which scales as $\log \log N$. In words, this means that we need to compute the Bhattacharyya parameter of a sublinear number of channels. This is possible only because we exploit the partial order of Proposition \ref{prop:order}.

\begin{figure*}
\begin{subfigure}[t]{\textwidth}
\centering
\setlength{\unitlength}{1.0bp}%
\begin{picture}(235,28)(0,0)
\put(-95,-32){\includegraphics[angle=90,scale=1.35]{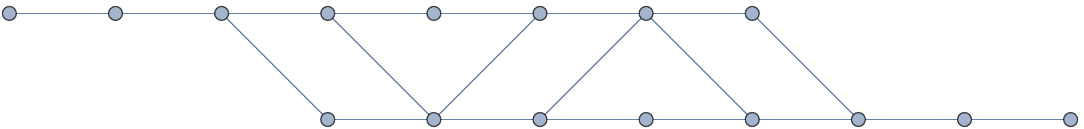}}
{
\scriptsize
\put (-110,18)  {$(0, 0, 0, 0)$}
\put (-67.5,2) {$(0, 0, 0, 1)$}
\put (-25,18) {$(0, 0, 1, 0)$}
\put (14.5, 2) {$(0, 1, 0, 0)$}
\put (56,18) {$(1, 0, 0, 0)$}
\put (97.5,2) {$(1, 0, 0, 1)$}
\put (139,18) {$(1, 0, 1, 0)$}
\put (180.5,2) {$(1, 1, 0, 0)$}
\put (14.5,-38)  {$(0, 0, 1, 1)$}
\put (56,-22) {$(0, 1, 0, 1)$}
\put (97.5,-38) {$(0, 1, 1, 0)$}
\put (139,-22) {$(0, 1, 1, 1)$}
\put (180.5,-38) {$(1, 0, 1, 1)$}
\put (222,-22) {$(1, 1, 0, 1)$}
\put (263.5,-38) {$(1, 1, 1, 0)$}
\put (305,-22) {$(1, 1, 1, 1)$}
}
\end{picture}
\vspace{4.5em}
\caption{$n=4$}\label{fig:hasse4}
\end{subfigure}
\vspace{7em}

\begin{subfigure}[t]{\textwidth}
\centering
\setlength{\unitlength}{1.0bp}%
\begin{picture}(235,28)(0,0)
\put(-128,0){\includegraphics[angle=90,scale=1.57]{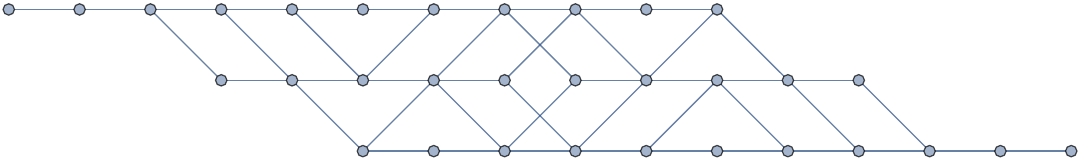}}
{
\scriptsize
\put (-144.5,73)  {$(0, 0, 0, 0, 0)$}
\put (-112.5,57)  {$(0, 0, 0, 0, 1)$}
\put (-80.5,73)  {$(0, 0, 0, 1, 0)$}
\put (-48.5,57)  {$(0, 0, 1, 0, 0)$}
\put (-16.5,73)  {$(0, 1, 0, 0, 0)$}
\put (15.5,57)  {$(1, 0, 0, 0, 0)$}
\put (47.5,73)  {$(1, 0, 0, 0, 1)$}
\put (79.5,57)  {$(1, 0, 0, 1, 0)$}
\put (111.5,73)  {$(1, 0, 1, 0, 0)$}
\put (143.5,57)  {$(1, 1, 0, 0, 0)$}
\put (175.5,73)  {$(1, 1, 0, 0, 1)$}
\put (-48.5,25.5)  {$(0, 0, 0, 1, 1)$}
\put (-16.5,41.5)  {$(0, 0, 1, 0, 1)$}
\put (15.5,25.5)  {$(0, 1, 0, 0, 1)$}
\put (47.5,41.5)  {$(0, 1, 0, 1, 0)$}
\put (79.5,25.5)  {$(0, 1, 1, 0, 0)$}
\put (111.5,41.5)  {$(1, 0, 0, 1, 1)$}
\put (143.5,25.5)  {$(1, 0, 1, 0, 1)$}
\put (175.5,41.5)  {$(1, 0, 1, 1, 0)$}
\put (207.5,25.5)  {$(1, 1, 0, 1, 0)$}
\put (239.5,41.5)  {$(1, 1, 1, 0, 0)$}
\put (15.5,-7)  {$(0, 0, 1, 1, 0)$}
\put (47.5,9)  {$(0, 0, 1, 1, 1)$}
\put (79.5,-7)  {$(0, 1, 0, 1, 1)$}
\put (111.5,9)  {$(0, 1, 1, 0, 1)$}
\put (143.5,-7)  {$(0, 1, 1, 1, 0)$}
\put (175.5,9)  {$(0, 1, 1, 1, 1)$}
\put (207.5,-7)  {$(1, 0, 1, 1, 1)$}
\put (239.5,9)  {$(1, 1, 0, 1, 1)$}
\put (271.5,-7)  {$(1, 1, 1, 0, 1)$}
\put (303.5,9)  {$(1, 1, 1, 1, 0)$}
\put (335.5,-7) {$(1, 1, 1, 1, 1)$}
}
\end{picture}
\vspace{1.5em}
\caption{$n=5$}\label{fig:hasse5}
\end{subfigure}
\caption{Hasse diagram of the partial order over the Hamming cube $\{0, 1\}^n$ induced by Proposition \ref{prop:order}.}\label{fig:hasse}
\end{figure*}
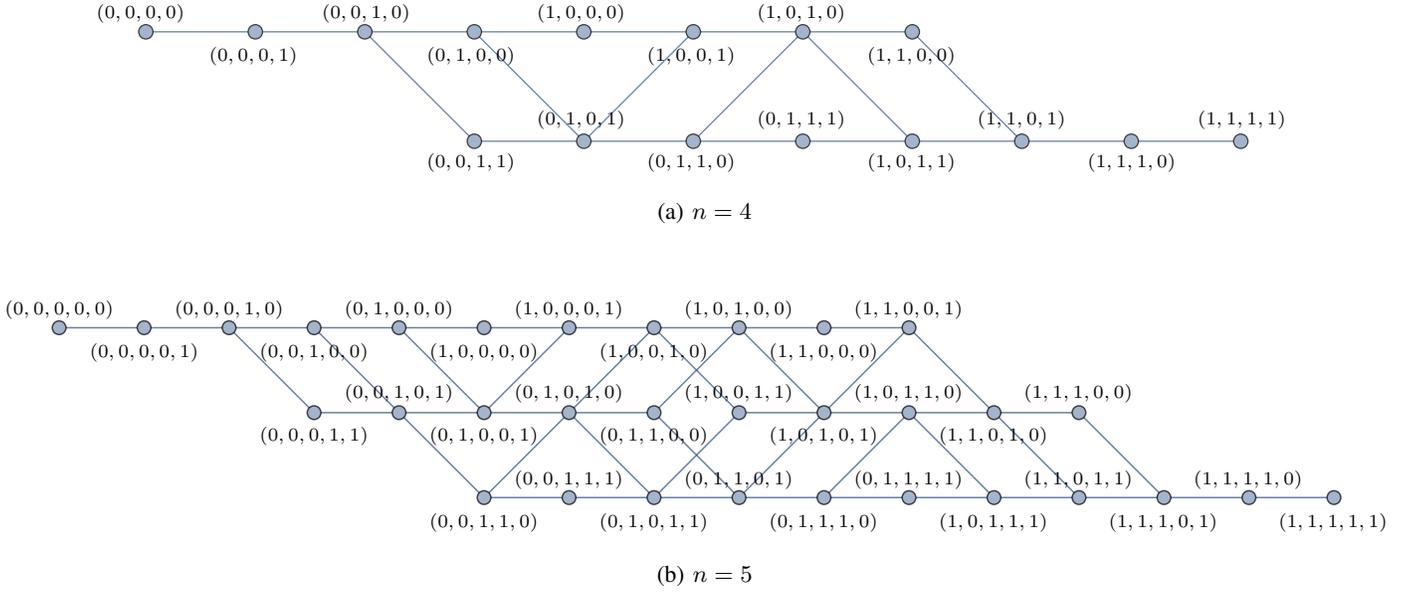

Indeed, assume that we do not use any partial order between the synthetic channels. Then, the only way to solve the FP construction problem is to compute the Bhattacharyya parameter of all the $N$ synthetic channels. On the contrary, suppose that there was a total order among the synthetic channels. Then, we could rank them from best to worst and, by using a binary search algorithm, we need to compute the Bhattacharyya parameter of at most $n+1=\log N +1$ synthetic channels. The main result of this paper is that by using the partial order of Proposition \ref{prop:order} we need to compute the Bhattacharyya parameter of roughly $N/\log^{3/2} N$ synthetic channels. Furthermore, as detailed at the end of Section \ref{sec:proof}, these $N/\log^{3/2} N$ synthetic channels can also be identified efficiently by solving a maximum matching problem on a bipartite graph.

Let us highlight that the bounds of Theorem \ref{th:main} hold when we exploit \emph{only} the partial order of Proposition \ref{prop:order}. This partial order relies on the addition and left-swap operators of Definitions \ref{def:add}-\ref{def:lw}, and these represent the only known operators that imply stochastic degradation. If one finds another operator that implies stochastic degradation, by exploiting the induced partial order, in principle it is possible to further reduce the number of Bhattacharyya parameters to be computed.

\section{Proof and Discussion} \label{sec:proof}

In order to prove Theorem \ref{th:main}, we need some definitions about partially ordered sets (or posets, for short). For a general introduction to the subject of posets, we refer the interested reader to \cite[Chapter 1]{E97} and \cite[Chapter 3]{Sta2011}. 
 
Let us associate the synthetic channel $W_N^{(i)}$ with the binary expansion $(i_1, i_2, \ldots, i_n)$ of the index $i$ defined in \eqref{eq:defbinexp}. Then, the partial order of Proposition \ref{prop:order} induces a partial order over the Hamming cube $\{0, 1\}^n$. We will denote such a partial order by $\prec_1$.

The \emph{Hasse diagram}  of the poset $\{0, 1\}^n$ equipped with the order $\prec_1$ is represented in Figure \ref{fig:hasse} for $n=4$ and $n=5$. Recall that an element $x$ is connected via an edge to an element $y$ {\em if and only if} they are ordered, i.e., $x \prec_1 y$ (respectively, $y \prec_1 x$) {\em and} there is no other element $z$ such that $x \prec_1 z \prec_1 y$ (respectively, $y \prec_1 z \prec_1 x$). In words, the Hasse diagram connects only ``nearest neighbors''.

Let us now define the concepts of chain and antichain that play a central role in our analysis. 

\begin{definition}[Chain and Antichain]\label{def:chain}
Let $P$ be a poset. We say that a subset of $P$ is a \emph{chain} if it is totally ordered. We say that a subset of $P$ is an \emph{antichain} if no two elements in it are comparable. 
\end{definition}

\begin{example}[Chain and Antichain]
Consider the partial order over $\{0, 1\}^4$ whose Hasse diagram is represented in Figure \ref{fig:hasse4}. Define 
\begin{equation*}
\begin{split}
\mathcal{C} &= \{(0, 0, 1, 0), (0, 0, 1, 1), (0, 1, 0, 1), (1, 0, 0, 1) \}, \\
\mathcal{A} &= \{(1, 0, 0, 0), (0, 1, 1, 1)\}.
\end{split}
\end{equation*}
Then, $\mathcal{C}$ is a chain and $\mathcal{A}$ is an antichain. Indeed, the elements $(1, 0, 0, 0)$ and $(0, 1, 1, 1)$ are not comparable and we have that
\begin{equation*}
(0, 0, 1, 0) \prec_1 (0, 0, 1, 1) \prec_1 (0, 1, 0, 1) \prec_1 (1, 0, 0, 1).
\end{equation*}
Analogously, consider the partial order over the set of synthetic channels $\{W_N^{(i)}\}_{i\in \{0, \ldots, N-1\}}$ given by Proposition \ref{prop:order}. Define
\begin{equation*}
\begin{split}
\mathcal{C}' &= \{W_{16}^{(2)}, W_{16}^{(3)}, W_{16}^{(5)}, W_{16}^{(9)}\}, \\
\mathcal{A}' &= \{W_{16}^{(8)}, W_{16}^{(7)}\}.
\end{split}
\end{equation*}
Then, $\mathcal{C}'$ is a chain and $\mathcal{A}'$ is an antichain. Indeed, the synthetic channels $W_{16}^{(8)}$ and $W_{16}^{(7)}$ are not comparable and we have that
\begin{equation*}
W_{16}^{(2)} \prec W_{16}^{(3)} \prec W_{16}^{(5)} \prec W_{16}^{(9)}.
\end{equation*}
\end{example}

The maximum cardinality of an antichain is equal to the minimum number of chains that form a partition of the poset by Dilworth's theorem \cite[Theorem 1.2]{Dil54}, \cite[Theorem 12.5]{BoM08}.

\begin{theorem}[Dilworth]
The minimum number of chains into which the elements of a poset $P$ can be partitioned is equal to the maximum number of elements in an antichain of $P$.
\end{theorem}

\begin{example}[Partition into Chains]
Consider the partial order over $\{0, 1\}^4$ whose Hasse diagram is represented in Figure \ref{fig:hasse4}. As the set is not totally ordered, we cannot find a chain that contains all its elements. However, we can find a partition of it into two chains. For example, we can pick the following two chains:
\begin{equation*}
\begin{split}
\mathcal{C}_1 = \{(0, 0, 0, 0), & (0, 0, 0, 1),  (0, 0, 1, 0), (0, 1, 0, 0), \\
(1, 0, 0, 0), & (1, 0, 0, 1), (1, 0, 1, 0), (1, 1, 0, 0)\} \\ 
\mathcal{C}_2 = \{ (0, 0, 1, 1), &  (0, 1, 0, 1), (0, 1, 1, 0), (0, 1, 1, 1), \\
(1, 0, 1, 1), & (1, 1, 0, 1),  (1, 1, 1, 0), (1, 1, 1, 1)\}.
\end{split}
\end{equation*}
Note that this decomposition is not unique, and there are other ways of partitioning the set $\{0, 1\}^n$ into two chains. As predicted by Dilworth's theorem, the maximum cardinality of an antichain is $2$. Indeed, $\mathcal{A} = \{(1, 0, 0, 0), (0, 1, 1, 1)\}$ is an antichain and it is easy to verify that there is no antichain of cardinality $3$.
\end{example}

The following lemma, whose proof is deferred to Appendix \ref{app:size}, characterizes the maximum cardinality of an antichain of the poset $\{0, 1\}^n$ with the order $\prec_1$.


\begin{lemma}[Maximum Cardinality of an Antichain]\label{lemma:sizeanti}
Let $M(n)$ be the maximal number of subsets of $\{1,\ldots,n\}$ that share the same sum. Consider the set $\{0, 1\}^n$ with the partial order $\prec_1$ and let $\mathcal A$ be an antichain. Then, 
\begin{equation}
\max_{\mathcal A}|\mathcal{A}| = M(n),
\end{equation}  
where the maximum is computed over the set of all antichains.
\end{lemma}

We are now ready to prove the main result of this paper.  

\begin{proof}[Proof of Theorem \ref{th:main}]
Consider the set of synthetic channels $\{W_N^{(i)}\}_{i\in \{0, \ldots, N-1\}}$ with the partial order given by Proposition \ref{prop:order}. Let $\mathcal C' \subseteq \{W_N^{(i)}\}$ be a chain. By Definition \ref{def:chain}, $\mathcal C'$ is totally ordered. Hence, in order to establish which elements of $\mathcal{C}'$ have a Bhattacharyya parameter smaller than $\gamma$, we can use a binary search algorithm, which requires the computation of at most $\lfloor \log |\mathcal C'|+1\rfloor$ Bhattacharyya parameters.

Let $(\mathcal{C}'_1, \ldots, \mathcal{C}'_K)$ be a partition of $\{W_N^{(i)}\}_{i\in \{0, \ldots, N-1\}}$ into a minimum number of chains. Clearly, the FP construction problem is equivalent to the problem of establishing which elements of $\mathcal{C}'_i$ have a Bhattacharyya parameter smaller than $\gamma$ for all $i\in \{1, \ldots, K\}$. In order to solve this last problem, the number of Bhattacharyya parameters to be computed is bounded as follows:
\begin{equation}\label{eq:proof1}
\begin{split}
\sum_{i=1}^K \left\lfloor\log |\mathcal C'_i|+1\right\rfloor &\le \sum_{i=1}^K \left(\log |\mathcal C'_i|+1\right) \\
&= K \cdot \left(1+\sum_{i=1}^K \frac{1}{K} \log |\mathcal C'_i|\right)\\
& \stackrel{\mathclap{\mbox{\footnotesize(a)}}}{\le} K \cdot \left(1+  \log \sum_{i=1}^K \frac{1}{K} |\mathcal C'_i|\right)\\
& = K \cdot \left(1+  \log \frac{2^n}{K} \right) = K \cdot \log \frac{2^{n+1}}{K},
\end{split}
\end{equation}
where the inequality (a) is an application of Jensen's inequality.

Let $\mathcal{A}'$ be an antichain. By Definition \ref{def:chain}, every pair of synthetic channels in $\mathcal{A}'$ is not comparable. Hence, in order to solve the FP construction problem, we necessarily need to compute the Bhattacharyya parameter of all the elements of $\mathcal{A}'$. By considering an antichain of maximum cardinality, we conclude that we need to compute at least
\begin{equation}\label{eq:proof2}
|\mathcal{A}'| = K
\end{equation} 
Bhattacharyya parameters, where the equality comes from Dilworth's theorem.

The set $\{W_N^{(i)}\}_{i\in \{0, \ldots, N-1\}}$ with the partial order given by Proposition \ref{prop:order} is order-isomorphic to the set $\{0, 1\}^n$ with the partial order $\prec_1$. Hence, by applying Lemma \ref{lemma:sizeanti}, we have that 
\begin{equation}\label{eq:proof3}
K = M(n).
\end{equation}
By combining \eqref{eq:proof1}, \eqref{eq:proof2}, and \eqref{eq:proof3}, the thesis immediately follows. 
\end{proof}

The result that we have just proved tightly bounds the \emph{number} of synthetic channels whose Bhattacharyya parameter has to be computed in order to solve the FP construction problem. Let us now describe \emph{how to find} these synthetic channels. 

Following the reasoning of the proof above, in order to solve the FP construction problem, we need to find a partition into chains of the set of synthetic channels. Then, for each chain, we establish which of the synthetic channels is reliable via a binary search algorithm. It remains to discuss how to find the partition into chains. 

Consider the bipartite graph $G = (U, V, E)$, where $U = V = \{0, 1\}^n$ and where $(u, v)$ is an edge in $G$ if and only if $u \prec_1 v$. The graph $G$ is represented in Figure \ref{fig:bipartite} for $n=4$. Recall that a matching is a set of edges without common vertices. Given a matching $M$ containing $m$ edges, we can associate to it the partition of $\{0, 1\}^n$ defined as follows: for each edge $(x,y)$ in $M$, include $x$ and $y$ in the same subset.

Suppose that $x$ and $y$ belong to the same subset. Then, there are two possibilities: either $(x,y)\in M$, which implies that $x\prec_1 y$; or there exists a set of intermediate vertices $z_1, \ldots, z_k$ such that $(x, z_1), (z_1, z_2), \ldots, (z_{k-1}, z_k), (z_k, y)\in M$, which implies that $x\prec_1 z_1 \prec_1 \cdots \prec_1 z_k \prec_1 y$. In both cases, $x$ and $y$ are comparable. Hence, $P$ is a partition of $\{0, 1\}^n$ into chains. Note also that the partition $P$ contains $|U|-|M|=2^n-m$ chains.

Similarly, given a partition $P$ of $\{0, 1\}^n$ into $p$ chains, we can associate to it the set of edges $M$ defined as follows. For $i\in \{1, \ldots, p\}$, let $P_i = \{x^{(i)}_1, \ldots, x^{(i)}_k\}$ be a chain of $P$. Then, we include the edges $(x^{(i)}_1, x^{(i)}_2), \ldots, (x^{(i)}_{k-1}, x^{(i)}_k)$ in $M$. Clearly, $M$ is a matching and it contains $|U|-|P|=2^n-p$ edges. 

In conclusion, we have described a way to associate the matchings of the graph $G$ to the partitions of $\{0, 1\}^n$ into chains and vice versa. Therefore, in order to find the partition of $\{0, 1\}^n$ containing the smallest number of chains, it suffices to find a maximum matching for the bipartite graph $G$. The last one is a classical problem in graph theory and it can be solved, e.g., via the Ford-Fulkerson algorithm in $O(|U|\cdot |E|) \le O(N^3)$ \cite{FF56} or via the Hopcroft-Karp algorithm in $O(\sqrt{|U|}\cdot |E|) \le O(N^{5/2})$ \cite{HK73}.

\begin{figure*}[t] 
\centering
\setlength{\unitlength}{1bp}%
\begin{picture}(0,0)(100,100)
\put(-210,120){\includegraphics[angle=270,scale=0.65]{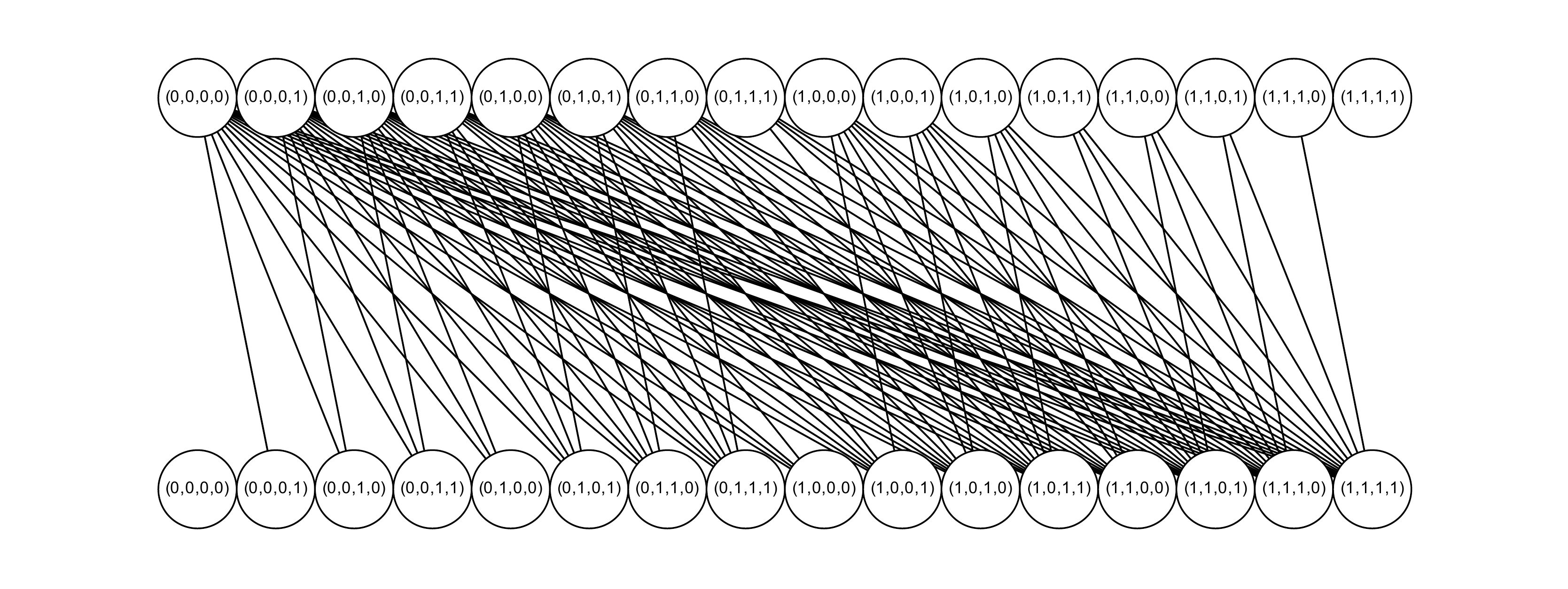}}
\end{picture}
\vspace{21em}
\caption{Bipartite graph $G = (U, V, E)$, where $U = V = \{0, 1\}^4$ and where $(u, v)$ is an edge in $G$ if and only if $u \prec_1 v$. \label{fig:bipartite}}
\end{figure*}

\section{Conclusions}\label{sec:concl}

In this work, we consider the problem of constructing a polar code of block length $N=2^n$ and we show that, by taking advantage of the partial order described in \cite{Sc16, BDOT16}, we need to compute the reliability of roughly a fraction $1/\log^{3/2} N$ of the synthetic channels. Note that this result holds regardless of the method used to compute the Bhattacharyya parameters (Monte Carlo simulations, Gaussian approximation, efficient degrading and upgrading, density evolution, and so on).

Let us briefly discuss the case of density evolution. In order to compute a single Bhattacharyya parameter, $\log N$ intermediate density evolution steps are necessary. However, the task of computing all the $N$ Bhattacharyya parameters can be implemented more efficiently since we can reuse some of these intermediate steps in the computation of multiple synthetic channels. In this way, instead of $N\log N$ density evolution steps, one needs to perform only $2N-1$ such steps. The main result of this work implies that we need roughly $N/\sqrt{\log N}$ density evolution steps, since we need to compute the reliability of $N/\log^{3/2} N$ synthetic channels and each of these computations requires $\log N$ intermediate density evolution steps.
 It remains an open problem to establish how much more we can save by reusing some of the intermediate steps.

The idea of the proof consists in relating the construction problem to the problem of computing the maximum cardinality of an antichain for a suitably defined poset. In particular, we prove that a lower bound to the number of synthetic channels whose reliability has to be computed is equal to the maximum cardinality of an antichain. Furthermore, this bound is tight up to a multiplicative factor scaling as $\log\log N$.  Eventually, we show that the maximum cardinality of an antichain for the poset taken into account is equal to the maximal number of subsets of $\{1, \ldots, n\}$ that share the same sum. Such a sequence is the integer sequence A025591 in \cite{OEISref} and it scales as $N/\log^{3/2} N$.

In order to establish which are the indices of these $N/\log^{3/2} N$ synthetic channels, we need to solve a maximum matching problem on a bipartite graph, which can be done in $O(N^{5/2})$. Note that this operation has to be performed only once, since by computing the reliability of those synthetic channels we can solve the construction problem for any BMS channel.

Let us point out that the main idea of the proof technique is completely general in the sense that it does not depend on the particular order described in \cite{Sc16, BDOT16}. Indeed, suppose that a new partial order on the synthetic channels of polar codes is found. Then, in order to improve the bounds on the the number of synthetic channels to be considered, it suffices to compute the maximum cardinality of an antichain for the poset induced by the new order.   

\section*{Acknowledgement}

The work of M.~Mondelli and R.~Urbanke was supported by grant No. 200021\_166106 of the Swiss National Science Foundation. M.~Mondelli was also supported by the Dan David Foundation.

\appendix

\subsection{Maximum Cardinality of Antichain: Proof of Lemma \ref{lemma:sizeanti}}\label{app:size}

Let $\mathcal{P}([n])$ denote the set of subsets of $\{1, \ldots, n\}$. Consider the following order relation: let $x = \{x_1, x_2, \ldots, x_k\}$ and $y = \{y_1, y_2, \ldots, y_j\}$ be elements of $\mathcal{P}([n])$ with $x_1 < x_2 < \cdots < x_k$ and $y_1 < y_2 < \cdots < y_j$; then we define $x \preceq_{2} y$ if and only if $k \le j$ and $x_i \le y_i$ for all $i\in \{1, \ldots, k\}$. In words, $x \preceq_{2} y$ if and only if $x$ has at most as many elements as $y$ and, by ordering them in an increasing fashion, the $i$-th element of $x$ is not larger than the $i$-th element of $y$. 

The following lemma proves that the set $\mathcal{P}([n])$ with the order $\prec_2$ and the set $\{0, 1\}^n$ with the order $\prec_1$ are essentially the same.

\begin{lemma}[Order-Isomorphism]\label{lemma:iso}
The set $\mathcal{P}([n])$ with the order $\prec_{2}$ is order-isomorphic to the set $\{0, 1\}^n$ with the order $\prec_{1}$.
\end{lemma}

\begin{proof}
By definition of order-isomorphism, in order to prove the claim, we need to find a bijective function $f$ from $\{0, 1\}^n$ to $\mathcal{P}([n])$ such that $x \prec_1 y$ if and only if $f(x)\prec_2 f(y)$ for every $x, y \in \{0, 1\}^n$.

Consider the function $f: \{0, 1\}^n \to \mathcal{P}([n])$ defined as follows. Given $x = (x_n, x_{n-1},\ldots, x_1) \in \{0, 1\}^n$, we have that $i \in f(x)$ if and only if $x_i=1$. In words, we associate to a sequence of $n$ bits the set of indices corresponding to the $1$s. Note that we index the sequence from right to left, i.e., the left-most bit of the sequence has index $n$ and the right-most bit of the sequence has index $1$.

Assume that $x \prec_1 y$. This means that $y$ is obtained from $x$ by applying addition and left-swap operators. Then, $x$ has at most as many $1$s as $y$, which implies that $f(x)$ has at most as many elements as $f(y)$. Furthermore, the $1$s of $y$ are placed more to the left than the $1$s of $x$, which implies that $i$-th element of $x$ is not larger than the $i$-th element of $y$ for all $i$. Hence, $f(x) \prec_2 f(y)$. Analogously, we can prove that $f(x) \prec_2 f(y)$ implies that $x \prec_1 y$, which yields the desired claim.
\end{proof}

Since $(\mathcal{P}([n]),\prec_2)$ is order-isomorphic to $(\{0, 1\}^n, \prec_1)$, it suffices to compute the maximum cardinality of an antichain of the former poset. To do so, let us define the concept of rank function.

\begin{definition}[Rank Function]\label{def:rank}
Given a poset $P$ with the order $\prec$, a \emph{rank function} is a function $\rho : P \to \mathbb N$ that fulfills the following properties:
\begin{enumerate}
\item if $x$ is a minimal\footnote{We say that $x$ is a \emph{minimal} element of $P$ if there is no $y\in P$ such that $y \prec x$.} element of $P$, then $\rho(x) =0$; 
\item if $y$ covers\footnote{We say that $y$ \emph{covers} $x$ if $x \prec y$ and there is no other element $z$ such that $x \prec z \prec y$.} $x$ , then $\rho(y) = \rho(x) + 1$. 
\end{enumerate}
\end{definition}

If a poset is equipped with a rank function $\rho$, we say that the element $x$ has rank $\rho(x)$. The set $\mathcal{P}([n])$ with the order $\prec_2$ is equipped with a rank function, as proved in the following lemma.

\begin{lemma}[Rank Function for ${\mathcal P}(${$[n]$}$), \prec_2)$] \label{lemma:rank}
Given $x = \{x_1, x_2, \ldots, x_k\} \in \mathcal{P}([n])$, define
\begin{equation}\label{eq:defrank}
\rho(x) = \sum_{i=1}^{k} x_i.
\end{equation}
Then, $\rho$ is the rank function for the set $\mathcal{P}([n])$ with the order $\prec_2$.
\end{lemma}

\begin{proof}
Clearly, $x=\emptyset$ is the unique minimal element of the poset $\mathcal{P}([n])$ with the order $\prec_2$. Furthermore, $\rho(x)=0$, which proves the first property of Definition \ref{def:rank}. 

Assume now that $y$ covers $x$. Then, either $y = x \cup \{1\}$ or $x$ and $y$ differ only in 1 element, say the $i$-th element, and $y_i = x_i +1$. In both these cases, $\rho(y) = \rho(x) + 1$, which proves the second property of Definition \ref{def:rank}.
\end{proof}

Let $P$ be a poset equipped with a rank function. If every maximal\footnote{We say that $x$ is a \emph{maximal} element of $P$ if there is no $y\in P$ such that $x \prec y$.} element has the same rank, call it $r_{\rm max}$, then we say that $P$ is a \emph{graded} poset and we can decompose it as
\begin{equation}
P = P_0 \cup P_1 \cup \cdots \cup P_{r_{\rm max}},
\end{equation}
where $P_i$ contains all the elements of $P$ with rank $i$. Every chain with the maximum cardinality passes through exactly one element of each of the subsets $P_i$, starting from $P_0$, then $P_1$, and so on.



Note that if two elements of a poset have the same rank, then they are not comparable. Therefore, for all $i\in \{1, \ldots, r_{\rm max}\}$, the subset $P_i$ is an antichain. The following definition relates these antichains to the antichain with the maximum cardinality.

\begin{definition}[Sperner Property]
Let $P$ be a graded poset and let $P_i$ be the antichain that contains all the elements of $P$ with rank $i$. We say that $P$ has the \emph{Sperner property} if the maximum cardinality of an antichain is equal to $\max_{i} |P_i|$.
\end{definition}

\begin{lemma}[Sperner Property for $({\mathcal P}(${$[n]$}$), \prec_2)$]\label{lemma:sperner}
The poset $\mathcal{P}([n])$ with the order $\prec_2$ has the Sperner property.
\end{lemma}

The proof of the result above is algebraic and it follows from Theorem 4.1 of \cite{St91} (see also \cite[Section 4.1.2]{St91} for a more detailed discussion). Eventually, we are ready to prove Lemma \ref{lemma:sizeanti}.

\begin{proof}[Proof of Lemma \ref{lemma:sizeanti}]
Consider the poset $\mathcal{P}([n])$ with the order $\prec_2$. By Lemma \ref{lemma:rank}, its rank function is given by \eqref{eq:defrank}. Furthermore, the maximal element is unique, hence the poset is graded. This maximal element is $\{1, 2, \ldots, n\}$ and it has rank
\begin{equation}\label{eq:rmax}
r_{\rm max} = \sum_{i=1}^{n} i = \frac{n(n+1)}{2}.
\end{equation}
Furthermore, by Lemma \ref{lemma:sperner}, $({\mathcal P}(${$[n]$}$), \prec_2)$ has the Sperner property. Hence, the cardinality of the largest antichain is given by
\begin{equation*}
\max_{i\in \{0, \ldots, r_{\rm max}\}} |P_i| = M(n),
\end{equation*} 
where $M(n)$ is defined as the maximal number of subsets $\{1, \ldots, n\}$ that share the same sum. Since $(\mathcal{P}([n]),\prec_2)$ is order-isomorphic to $(\{0, 1\}^n, \prec_1)$ by Lemma \ref{lemma:iso}, the thesis immediately follows.
\end{proof}

As a final remark, let us point out other two interesting properties of $(\mathcal{P}([n]), \prec_2)$ that follow from the discussion in \cite[Section 4.1.2]{St91}:
\begin{itemize}
\item The poset $\mathcal{P}([n])$ with the order $\prec_2$ is \emph{rank symmetric}, i.e., $|P_i| = |P_{r_{\rm max}-i}|$ for all $i\in \{0, \ldots, r_{\rm max}\}$, where $r_{\rm max}$ is given by \eqref{eq:rmax}.
\item  The poset $\mathcal{P}([n])$ with the order $\prec_2$ is \emph{rank unimodal}, i.e., there is a $j$ such that $|P_0| \le |P_1| \le \cdots \le |P_j| \ge |P_{j+1}| \ge \cdots \ge |P_{r_{\rm max}}|$.
\end{itemize}

As a result, we have that 
\begin{equation*}
\max_{i\in \{0, \ldots, r_{\rm max}\}} |P_i| = |P_{\lfloor r_{\rm max}/2 \rfloor}| = |P_{\lceil r_{\rm max}/2 \rceil}|.
\end{equation*}
In words, the subset(s) $P_i$ with the maximum cardinality correspond to the middle rank(s). This means that $M(n)$ is equal to the number of subsets of $\{1, \ldots, n\}$ that have sum equal to $\lfloor n(n+1)/4\rfloor$ or to $\lceil n(n+1)/4\rceil$.

\subsection{Asymptotic Formula for $M(n)$: Proof of Theorem \ref{th:asym}}\label{app:asym}

\begin{proof}[Proof of Theorem \ref{th:asym}]
Recall that $M(n)$ is defined as the maximal number of subsets of $\{1,\ldots,n\}$ that share the same sum. Clearly, for any $K\in \mathbb N$, the number of subsets of $\{1,\ldots,n+1\}$ with sum $K$ is no smaller than the number of subsets of $\{1,\ldots,n\}$ with sum $K$. Hence, $M(n)$ is an increasing sequence and its limit is equal to the limit of any of its subsequences. The rest of the proof consists in showing that a suitably defined subsequence of $M(n)$ has the asymptotic behavior given by \eqref{eq:asympt}.

From the discussion at the end of Appendix \ref{app:size}, we have that the integer $K$ that maximizes the number of subsets of $\{1,\ldots,n\}$ with sum $K$ is $\lfloor n(n+1)/4\rfloor$ or $\lceil n(n+1)/4\rceil$. Assume now that $n\equiv 0$ or $n\equiv 3$ modulo $4$. Then, we have that
\begin{equation*}
\frac{n(n+1)}{4} \in \mathbb N.
\end{equation*}
Furthermore, we claim that $M(n)$ is equal to the number of choices of $+$ and $-$ signs such that
\begin{equation}\label{eq:pm}
\pm 1 \pm 2 \ldots \pm n=0.
\end{equation}
To see this, let $A$ be a subset of $\{1,\ldots,n\}$ with sum $n(n+1)/4$. Then, the set $\{1,\ldots,n\} \setminus A$ has also sum $n(n+1)/4$. By associating the positive sign to the elements of $A$ and the negative sign to the elements of $\{1,\ldots,n\} \setminus A$, we have that the overall sum is $0$. As a result, we have found a bijection between the set of subsets of $\{1,\ldots,n\}$ with sum $n(n+1)/4$ and the set of choices of $+$ and $-$ signs such that \eqref{eq:pm} holds. 

Define $S(n)$ as the number of choices of $+$ and $-$ signs such that \eqref{eq:pm} holds. From the discussion in the previous paragraph, we conclude that $S(n)=M(n)$ for $n\equiv 0$ or $n\equiv 3$ modulo $4$. Thus, the asymptotic behavior of $S(n)$ for $n\equiv 0$ or $n\equiv 3$ modulo $4$ is the same as the asymptotic behavior of $M(n)$.

In \cite[Theorem 2.1]{AnTo02}, it is proved that $S(n)$ is equal to the coefficient of $x^{n(n+1)/4}$ in the expansion of $\prod_{i=1}^n (1+x^i)$ and it is it conjectured that 
\begin{equation*}
S(n) = \sqrt{\frac{6}{\pi}}\frac{2^n}{n^{3/2}}(1 +o(1)).
\end{equation*}
This conjecture is proved in \cite{Su13}, which implies our desired result. 
\end{proof}

\bibliographystyle{IEEEtran}
\bibliography{lth,lthpub}

\newcommand{\SortNoop}[1]{}
\begin{thebibliography}{10}
\providecommand{\url}[1]{#1}
\csname url@samestyle\endcsname
\providecommand{\newblock}{\relax}
\providecommand{\bibinfo}[2]{#2}
\providecommand{\BIBentrySTDinterwordspacing}{\spaceskip=0pt\relax}
\providecommand{\BIBentryALTinterwordstretchfactor}{4}
\providecommand{\BIBentryALTinterwordspacing}{\spaceskip=\fontdimen2\font plus
\BIBentryALTinterwordstretchfactor\fontdimen3\font minus
  \fontdimen4\font\relax}
\providecommand{\BIBforeignlanguage}[2]{{%
\expandafter\ifx\csname l@#1\endcsname\relax
\typeout{** WARNING: IEEEtran.bst: No hyphenation pattern has been}%
\typeout{** loaded for the language `#1'. Using the pattern for}%
\typeout{** the default language instead.}%
\else
\language=\csname l@#1\endcsname
\fi
#2}}
\providecommand{\BIBdecl}{\relax}
\BIBdecl

\bibitem{Sc16}
C.~Sch{\"u}rch, ``A partial order for the synthesized channels of a polar
  code,'' in \emph{Proc. of the IEEE Int. Symposium on Inform. Theory (ISIT)},
  Barcelona, Spain, July 2016, pp. 220--224.

\bibitem{BDOT16}
M.~Bardet, V.~Dragoi, A.~Otmani, and J.-P. Tillich, ``Algebraic properties of
  polar codes from a new polynomial formalism,'' in \emph{Proc. of the IEEE
  Int. Symposium on Inform. Theory (ISIT)}, Barcelona, Spain, July 2016, pp.
  230--234.

\bibitem{Ari09}
E.~{Ar\i kan}, ``Channel polarization: A method for constructing
  capacity-achieving codes for symmetric binary-input memoryless channels,''
  \emph{IEEE Trans. Inform. Theory}, vol.~55, no.~7, pp. 3051--3073, July 2009.

\bibitem{MHU15unif-ieeeit}
M.~Mondelli, S.~H. Hassani, and R.~Urbanke, ``Unified scaling of polar codes:
  {E}rror exponent, scaling exponent, moderate deviations, and error floors,''
  \emph{IEEE Trans. Inform. Theory}, vol.~62, no.~12, pp. 6698--6712, Dec.
  2016.

\bibitem{ArT09}
E.~{Ar\i kan} and I.~E. {Telatar}, ``{On the rate of channel polarization},''
  in \emph{Proc. of the IEEE Int. Symposium on Inform. Theory (ISIT)}, Seoul,
  South Korea, July 2009, pp. 1493--1495.

\bibitem{HAU14}
S.~H. Hassani, K.~Alishahi, and R.~Urbanke, ``Finite-length scaling for polar
  codes,'' \emph{IEEE Trans. Inform. Theory}, vol.~60, no.~10, pp. 5875--5898,
  Oct. 2014.

\bibitem{GB14}
D.~Goldin and D.~Burshtein, ``Improved bounds on the finite length scaling of
  polar codes,'' \emph{IEEE Trans. Inform. Theory}, vol.~60, no.~11, pp.
  6966--6978, Nov. 2014.

\bibitem{TVa15}
I.~Tal and A.~Vardy, ``{List decoding of polar codes},'' \emph{IEEE Trans.
  Inform. Theory}, vol.~61, no.~5, pp. 2213--2226, May 2015.

\bibitem{MoT09}
R.~Mori and T.~Tanaka, ``Performance and construction of polar codes on
  symmetric binary-input memoryless channels,'' in \emph{Proc. of the IEEE Int.
  Symposium on Inform. Theory (ISIT)}, Seoul, South Korea, July 2009, pp.
  1496--1500.

\bibitem{MoTCommLett09}
------, ``Performance of polar codes with the construction using density
  evolution,'' \emph{IEEE Commun. Lett.}, vol.~13, no.~7, pp. 519--521, July
  2009.

\bibitem{Tr12}
P.~Trifonov, ``Efficient design and decoding of polar codes,'' \emph{IEEE
  Trans. Commun.}, vol.~60, no.~11, pp. 3221--3227, Nov. 2012.

\bibitem{TV13con}
I.~Tal and A.~Vardy, ``How to construct polar codes,'' \emph{IEEE Trans.
  Inform. Theory}, vol.~59, no.~10, pp. 6562--6582, Oct. 2013.

\bibitem{RHTT}
R.~Pedarsani, H.~Hassani, I.~Tal, and E.~Telatar, ``On the construction of
  polar codes,'' in \emph{Proc. of the IEEE Int. Symposium on Inform. Theory
  (ISIT)}, St. Petersberg, Russia, Aug. 2011, pp. 11--15.

\bibitem{OEISref}
N.~J.~A. Sloane, ``{The On-Line Encyclopedia of Integer Sequences},'' published
  electronically at \texttt{https://oeis.org}.

\bibitem{RiU08}
T.~Richardson and R.~Urbanke, \emph{Modern Coding Theory}.\hskip 1em plus 0.5em
  minus 0.4em\relax Cambridge University Press, 2008.

\bibitem{E97}
K.~Engel, \emph{Sperner Theory}.\hskip 1em plus 0.5em minus 0.4em\relax
  Cambridge University Press, 1997.

\bibitem{Sta2011}
R.~P. Stanley, \emph{Enumerative Combinatorics}, 2nd~ed., ser. Cambridge
  Studies in Advanced Mathematics.\hskip 1em plus 0.5em minus 0.4em\relax
  Cambridge University Press, 2011, vol.~1.

\bibitem{Dil54}
R.~P. Dilworth, ``A decomposition theorem for partially ordered sets,''
  \emph{Annals Math.}, vol.~51, no.~1, pp. 161--166, 1950.

\bibitem{BoM08}
{J. A. Bondy} and {U. S. R. Murty}, \emph{Graph Theory}.\hskip 1em plus 0.5em
  minus 0.4em\relax Springer, 2008.

\bibitem{FF56}
L.~R. Ford and D.~R. Fulkerson, ``Maximal flow through a network,''
  \emph{Canadian Journal of Mathematics}, vol.~8, no.~3, pp. 399--404, 1956.

\bibitem{HK73}
J.~E. Hopcroft and R.~M. Karp, ``An $n^{5/2}$ algorithm for maximum matchings
  in bipartite graphs,'' \emph{SIAM Journal on Computing}, vol.~2, no.~4, pp.
  225--231, 1973.

\bibitem{St91}
R.~P. Stanley, ``Some applications of algebra to combinatorics,''
  \emph{Discrete Applied Math.}, vol.~34, pp. 241--277, 1991.

\bibitem{AnTo02}
D.~Andrica and I.~Tomescu, ``On an integer sequence related to a product of
  trigonometric functions, and its combinatorial relevance,'' \emph{Journal of
  Integer Sequences}, vol.~5, 2002, {A}rticle 02.2.4.

\bibitem{Su13}
B.~D. Sullivan, ``On a conjecture of {A}ndrica and {T}omescu,'' \emph{Journal
  of Integer Sequences}, vol.~16, 2013, {A}rticle 13.3.1.

\end{thebibliography}

\end{document}